\documentclass[11pt]{article}
\usepackage{amsmath, amsthm, amssymb, mathtools}
\usepackage{amsfonts}
\usepackage{amssymb}
\usepackage{bbm}
\usepackage{graphicx}
\usepackage{caption}
\usepackage{fullpage}
\usepackage{microtype}
\usepackage{hyperref}

\hypersetup{
    colorlinks=true,
    linkcolor=blue,
    filecolor=magenta,      
    urlcolor=cyan,
}

\newtheorem{theorem}{Theorem}

\newtheorem{corollary}[theorem]{Corollary}

\newtheorem{lemma}{Lemma}

\newtheorem{definition}[theorem]{Definition}

\newcommand{\ket}[1]{|#1\rangle}
\newcommand{\bra}[1]{\langle#1|}
\newcommand{\ketbra}[2]{\ket{#1}\!\bra{#2}} 
\newcommand{\norm}[1]{\left\|#1\right\|}
\newcommand{\pnorm}[2]{\left\|#2\right\|_#1}
\newcommand{\Tr}{\mbox{Tr}}
\newcommand{\polylog}{\mbox{polylog}}
\newcommand{\poly}{\mbox{poly}}
\newcommand{\jj}{\hat{\jmath}}

\begin{document}

\title{Efficient Quantum Algorithms for Simulating Lindblad Evolution\thanks{A preliminary version of this article appeared in {\it Proceedings of the 44th International Colloquium on Automata, Languages, and Programming (ICALP 2017)}, pages 17:1--17:14.}}

\author{Richard Cleve%
\thanks{Institute for Quantum Computing, University of Waterloo, Waterloo, Canada.}
\thanks{Cheriton School of Computer Science, University of Waterloo, Waterloo, Canada.}
\thanks{Canadian Institute for Advanced Research, Toronto, Canada.}
\qquad Chunhao Wang$^{\dag\ddag}$
}

\date{\empty}

\maketitle

\begin{abstract}
We consider the natural generalization of the Schr\"{o}dinger equation to \textit{Markovian open system dynamics}: the so-called the Lindblad equation.
We give a quantum algorithm for simulating the evolution of an $n$-qubit system for time $t$ within precision $\epsilon$.
If the Lindbladian consists of $\mathrm{poly}(n)$ operators that can each be expressed as a linear combination of $\mathrm{poly}(n)$ tensor products of Pauli operators then the gate cost of our algorithm is 
$O(t\, \mathrm{polylog}(t/\epsilon)\mathrm{poly}(n))$.
We also obtain similar bounds for the cases where the Lindbladian consists of local operators, and where the Lindbladian consists of sparse operators.
This is remarkable in light of evidence that we provide indicating that the above efficiency is impossible to attain by first expressing Lindblad evolution as Schr\"{o}dinger evolution on a larger system and tracing out the ancillary system: the cost of such a \textit{reduction} incurs an efficiency overhead of $O(t^2/\epsilon)$ even before the Hamiltonian evolution simulation begins.
Instead, the approach of our algorithm is to use a novel variation of the ``linear combinations of unitaries'' construction that pertains to channels.
\end{abstract}

\section{Introduction}

The problem of simulating the evolution of closed systems (captured by the Schr\"odinger equation) was proposed by Feynman~\cite{Fey82} in 1982 as a motivation for building quantum computers.
Since then, several quantum algorithms have appeared for this problem (see subsection~\ref{sec:previous} for references to these algorithms).
However, many quantum systems of interest are not closed but are well-captured by the Lindblad Master equation~\cite{lindblad1976,GKS1976}.
Examples exist in quantum physics \cite{leggett1987, weiss2012}, quantum chemistry \cite{may2008, nitzan2006}, and quantum biology \cite{dorner2012, huelga2013, mostame2012}.
Lindblad evolution also arises in quantum computing and quantum information in the context of entanglement preparation~\cite{kraus2008, kastoryano2011, reiter2016}, thermal state preparation~\cite{kastoryano2016}, quantum state engineering~\cite{verstraete2009}, and studying the noise of quantum circuits~\cite{magesan2013}.

We consider the computational cost of simulating the evolution of an $n$-qubit quantum state for time $t$ under the Lindblad Master equation
\begin{align}\label{eq:lindblad}
  \dot{\rho} = -i[H,\rho] + \sum_{j=1}^{m}\Bigl(L_j\rho L_j^{\dag} 
- \frac{1}{2}L_j^{\dag}L_j\rho - \frac{1}{2}\rho L_j^{\dag}L_j\Bigr),
\end{align}
(representing Markovian open system dynamics), where $H$ is a Hamiltonian and 
$L_1, \dots, L_{m}$ are linear operators.
By \textit{simulate the evolution}, we mean: provide a quantum circuit that computes the quantum channel corresponding to evolution by Eq.~\eqref{eq:lindblad} for time $t$ within precision $\epsilon$.
The quantum circuit must be independent of the input state, which is presumed to be unknown.
When $L_1 = \cdots = L_m = 0$, Eq.~\eqref{eq:lindblad} is the Schr\"odinger equation.

Eq.~\eqref{eq:lindblad} can be viewed as an idealization of the frequently occurring physical scenario where a quantum system evolves jointly with a large external environment in a manner where information dissipates from the system into the environment.
In quantum information theoretic terms, Lindblad evolution is a continuous-time process that, for any evolution time, is a quantum channel. 
Moreover, Lindblad evolution is \emph{Markovian} in the sense that, given the state at time $t$, for any 
$\delta > 0$, the state at time $t+\delta$ is a function of the state at time $t$ alone (i.e., is independent of the state before time~$t$).

Lindblad evolution can be intuitively thought of as Hamiltonian evolution in a larger system that includes an ancilla register, but where the ancilla register is being continually reset to its initial state.
To make this more precise, consider a time interval $[0,t]$, and divide it into $N$ subintervals of length $\frac{t}{N}$ each.
At the beginning of each subinterval, reset the state of the ancilla register to its initial state, and then let the joint system-ancilla evolve under a Hamiltonian $J$ and the system itself evolve under $H$.
Let the evolution time for $J$ be $\sqrt{t/N}$ and the evolution time for $H$ be $t/N$.
This process, illustrated in Fig.~\ref{fig:lindblad-as-unitary}, converges to the true Lindblad evolution as $N$ approaches $\infty$.
\begin{figure}[!ht]
	\centering
	\includegraphics[width=0.95\textwidth]{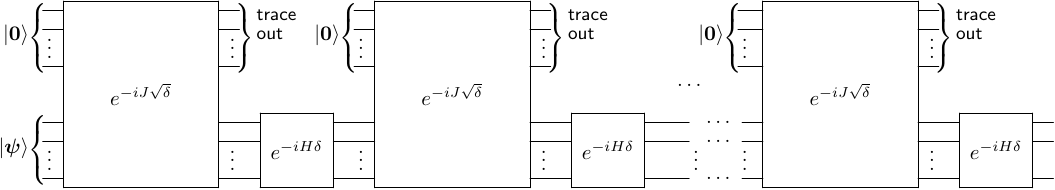}
	\caption{\small Lindblad evolution for time $t$ approximated by unitary operations.
	There are $N$ iterations and $\delta = t/N$. 
	This converges to Lindblad evolution as $N \rightarrow \infty$.	\label{fig:lindblad-as-unitary}}
\end{figure}
For the specific evolution described by Eq.~\eqref{eq:lindblad}, it suffices to set the ancilla register to $\mathbb{C}^{m+1}$ and the Hamiltonian $J$ to the block matrix
\begin{align}
J = 
\begin{pmatrix}
0 & L_1^{\dag} & \cdots & L_m^{\dag} \\
L_1 & 0 & \cdots & 0 \\
\vdots & \vdots & \ddots & \vdots \\
L_m & 0 & \cdots & 0 
\end{pmatrix}.
\end{align}

A remarkable property of this way of representing Lindblad evolution is that the rate at which the Hamiltonian $J$ evolves is effectively infinite:
Lindblad evolution for time $t/N$ is simulated by a process that includes evolution by $J$ for time $\sqrt{t/N}$, so the rate of the evolution scales as
\begin{align}
\frac{\sqrt{t/N}}{t/N} = \sqrt{\frac{N}{t}},
\end{align}
which diverges as $N \rightarrow \infty$.
Moreover, the total Hamiltonian evolution time of $J$ in Fig.~\ref{fig:lindblad-as-unitary} is $N \sqrt{t/N} = \sqrt{Nt}$, which also diverges.
In Appendix~\ref{appendix:total-is-infinite} we prove that, in general, the above scaling phenomenon is necessary for simulating time-independent Lindblad evolution in terms of time-independent Hamiltonian evolution along the lines of the overall structure of Fig.~\ref{fig:lindblad-as-unitary}.
In this sense, \emph{exact} Lindblad evolution for finite time does not directly correspond to Hamiltonian evolution for \emph{any} finite time.
On the other hand, it can be shown that if the scaling of $N$ is at least $t^3/\epsilon^2$ then the final state is an \textit{approximation} within $\epsilon$.
Note that then the corresponding total evolution time for $J$ scales as $\sqrt{(t^3/\epsilon^2)t} = t^2/\epsilon$.
Therefore, quantum algorithms that simulate Lindblad evolution by first applying the above reduction to Hamiltonian evolution and then efficiently simulating the Hamiltonian evolution are likely to incur scaling that is at least $t^2/\epsilon$.

Here we are interested in whether much more efficient simulations of Lindblad evolution are possible, such as $O(t\,\polylog(t/\epsilon))$.

\subsection{Previous work}\label{sec:previous}

\noindent\textbf{Simulating Hamiltonian evolution.}
Hamiltonian evolution (a.k.a.\ Schr\"odinger evolution) is the special case of Eq.~\eqref{eq:lindblad} where $L_j = 0$ for all $j$.
This simulation problem has received considerable attention since 
Feynman~\cite{Fey82} proposed this as a motivation for building quantum computers; see for example \cite{Llo96,AT03,Childs2004,BCCKS14,BCCKS15,BCG14,BCK15,kothari2014efficient,LowChuang16,PP16,BN2016}.
Some of the recent methods obtain a scaling that is 
$O(t\,\polylog(t/\epsilon) \mathrm{poly}(n))$, 
thereby exceeding what can be accomplished by the longstanding Trotter-Suzuki methods \cite{Suz91}.

\medskip

\noindent\textbf{Simulating Lindblad evolution.}
The natural generalization from closed systems to Markovian open systems in terms of the Lindblad equation has received much less attention.
Kliesch \textit{et al.}~\cite{kliesch2011dissipative} give a quantum algorithm for simulating Lindblad evolution in the case where each of $H, L_1, \dots, L_m$ can be expressed as a sum of local operators (i.e., which act on a constant number of qubits).
The cost of this algorithm with respect to $t$ and $\epsilon$ (omitting factors of $\mathrm{poly}(n)$) is $O(t^2/\epsilon)$.
In~\cite{CL16}, Childs and Li improve this to $O(t^{1.5}/\sqrt{\epsilon})$ 
and also give an $O((t^2/\epsilon)\polylog(t/\epsilon))$ query algorithm for the case where the operators in Eq.~\eqref{eq:lindblad} are sparse and represented in terms of an oracle.
Another result in \cite{CL16} is an $\Omega(t)$ lower bound for the query complexity for time $t$ when Eq.~\eqref{eq:lindblad} has $H = 0$ and $m=1$.

As far as we know, none of the previous algorithms for simulating Lindblad evolution has cost $O(t\,\polylog(t/\epsilon) \mathrm{poly}(n))$,
which is the performance that we attain.
Our results are summarized precisely in the next subsection (subsection~\ref{sec:new-result}).

We note that there are simulation algorithms that solve problems that are related to but different from ours, such as~\cite{DPCSC15}, which does not produce the final state; rather it simulates the expectation of an observable applied to the final state.
We do not know how to adapt these techniques to produce the unmeasured final state instead.

Finally, we note that there are interesting classical algorithmic techniques for simulating Lindblad evolution that are feasible when the dimension of the Hilbert space (which is $2^n$, for $n$ qubits) is not too large---but these do not carry over to the context of quantum algorithms (where $n$ can be large).
In the classical setting, since the state is known (and stored) explicitly, various ``unravellings'' of the process that are state-dependent can be simulated.
For example, the random variable corresponding to ``the next jump time'' (which is highly state-dependent) can be simulated.
In the context of quantum algorithms, the input state is unknown and cannot be measured without affecting it.

\subsection{New results}\label{sec:new-result}

Eq.~\eqref{eq:lindblad} can be written as $\dot{\rho} = \mathcal{L}[\rho]$,
where $\mathcal{L}$ is a \emph{Lindbladian}, defined as a mapping of the form
\begin{align}\label{eq:lindbladian}
\mathcal{L}[\rho] = -i[H,\ \rho] + \sum_{j=1}^{m}\Bigl(L_j\rho L_j^{\dag} 
- \frac{1}{2}L_j^{\dag}L_j\rho - \frac{1}{2}\rho L_j^{\dag}L_j\Bigr),
\end{align}
for operators $H,L_1, \dots, L_{m}$ on the Hilbert space $\mathcal{H}=\mathbb{C}^{2^n}$ ($n$ qubits) with $H$ Hermitian.
Evolution under Eq.~\eqref{eq:lindblad} for time $t$ corresponds to the quantum map $e^{\mathcal{L}t}$ (which is a channel for any $t \ge 0$).

Each of the operators $H,L_1, \dots, L_{m}$ corresponds to a $2^n \times 2^n$ matrix.
The simulation algorithm is based on a succinct specification of these matrices.
Our succinct specification is as a \emph{linear combination of $q$ Paulis}, defined as
\begin{align}\label{eq:H-lcu}
  H =& \sum_{k=0}^{q-1}\beta_{0k}V_{0k} \\
  \label{eq:jump-lcu}
  L_j =& \sum_{k=0}^{q-1}\beta_{jk}V_{jk},
\end{align}
where, for each $j\in\{0,\ldots,m\}$ and $k\in\{0,\ldots,q-1\}$,  $V_{jk}$ is an $n$-fold tensor product of Paulis ($I$, $\sigma_x$, $\sigma_y$, $\sigma_z$) and a scalar phase $e^{i\theta}$ ($\theta \in [0,2\pi]$), and $\beta_{jk} \ge 0$.

In the evolution $e^{\mathcal{L}t}$, it is possible to scale up
$\mathcal{L}$ by some factor while reducing $t$ by the same factor, i.e.,
$e^{\mathcal{L}t}[\rho] = e^{(c\mathcal{L})\frac{t}{c}}[\rho]$ for any $c>0$\footnote{$c\mathcal{L}$ denotes the mapping obtained from $\mathcal{L}$ with $H$ multiplied by $c$ and each $L_j$ multiplied by $\sqrt{c}$.}.
This reduces the simulation time but transfers the cost into the magnitude
of $\mathcal{L}$.
To normalize this cost, we define a norm 
based on the specification of $\mathcal{L}$.

Define the norm%
\footnote{For simplicity we use the terminology $\|\mathcal{L}\|_{\textsf{pauli}}$ even though the quantity is not directly a function of the mapping~$\mathcal{L}$.
However, $\|c\mathcal{L}\|_{\textsf{pauli}} = c\|\mathcal{L}\|_{\textsf{pauli}}$ if $c \mathcal{L}$ denotes the expression in Eq.~\eqref{eq:lindbladian} with the factor $c$ multiplied through.}
of a specification of a Lindbladian $\mathcal{L}$ as a linear product of Paulis as
\begin{align}\label{cost-lindbladian}
\|\mathcal{L}\|_{\textsf{pauli}} = \sum_{k=0}^{q-1}\beta_{0k} + \sum_{j=1}^m\Bigl(\,\sum_{k=0}^{q-1}\beta_{jk}\Bigr)^2.
\end{align}

Our main result is the following theorem.
\begin{theorem} \label{thm:main_thm}
Let $\mathcal{L}$ be a Lindbladian 
presented as a linear combination of $q$ Paulis.
Then, for any $t > 0$ and $\epsilon > 0$, there exists a quantum circuit of size 
  \begin{align}
    O\left(m^2q^2 \tau \frac{(\log(m q \tau/\epsilon)+n)\log(\tau/\epsilon)}{\log\log(\tau/\epsilon)}\right)
  \end{align}
that implements a quantum channel $\mathcal{N}$, such that
$\pnorm{\diamond}{\mathcal{N} - e^{\mathcal{L}t}} \le \epsilon$,
where $\tau = t\,\|\mathcal{L}\|_{\textrm{\sf pauli}}$ and $m$ is the number of jump operators in $\mathcal{L}$.
\end{theorem}

\noindent
\textbf{Remarks:}
\begin{enumerate}
\item
  The proof of Theorem~\ref{thm:main_thm} is in section~\ref{sec:main_thm}.
A main novel ingredient of the proof is Lemma~\ref{lemma:lcu}, concerning a variant of the ``linear combination of unitaries'' construction that is suitable for channels (explained in sections~\ref{sec:novel} and~\ref{sec:new-lcu}).
\item
The factor $\|\mathcal{L}\|_{\textrm{\sf pauli}}$ corresponding to the coefficients of the specification as a linear combination of Paulis is a natural generalization to the case of Lindbladians of a similar factor for Hamiltonians that appears in~\cite{BCCKS15}.
\item
  When $m, q \in \mathrm{poly}(n)$, the gate complexity in Theorem~\ref{thm:main_thm}
simplifies to
\begin{align}
O\left(\tau\,\frac{\log(\tau/\epsilon)^2}{\log\log(\tau/\epsilon)}\,\mathrm{poly}(n)\right).
\end{align}
\item
A Lindbladian $\mathcal{L}$ is \emph{local} if 
\begin{align}
H = \sum_{j=1}^{m'}H_j,
\end{align}
where $H_1,\dots, H_{m'}$ and also $L_1,\dots,L_m$ are local (i.e., they each act on a constant number of qubits).
A \emph{local specification} of $\mathcal{L}$ is as 
$H_1,\dots, H_{m'},L_1,\dots,L_m$ and we define its norm as
\begin{align}\label{eq:local-norm}
\|\mathcal{L}\|_{\mathsf{local}} = 
\sum_{j=1}^{m'}\|H_j\| + \sum_{j=1}^{m}\|L_j\|^2.
\end{align}
For local Lindbladians, Theorem~\ref{thm:main_thm} reduces to the following.
\begin{corollary}
If $\mathcal{L}$ is a local Lindbladian then the gate complexity for simulating  $e^{\mathcal{L}t}$ with precision $\epsilon$ is
	\begin{align}
	  O\left((m+m')^2\,\tau\,\frac{\log((m+m')\tau/\epsilon)\log(\tau/\epsilon)}{\log\log(\tau/\epsilon)}\right),
	\end{align}
where $\tau = t\,\|\mathcal{L}\|_{\mathsf{local}}$.
\end{corollary}
\item
We also consider \emph{sparse} Lindbladians (see \cite{CL16} for various definitions, extending definitions and specifications of sparse Hamiltonians \cite{AT03}).
Here, we define a Lindbladian to have $d$-sparse operators if $H,L_1,\dots,L_m$ each have at most $d$ non-zero entries in each row/column.
A \emph{sparse specification} of such a Lindbladian $\mathcal{L}$ is as a black-box that provides the positions and values of the non-zero entries of each row/column of $H,L_1,\dots,L_m$ via queries.

Define the norm of any specification of a Lindbladian in terms of operators $H,L_1,\dots,L_m$ as
\begin{align}\label{eq:sparse-norm}
\|\mathcal{L}\|_{\mathsf{ops}} = 
\|H\| + \sum_{j=1}^{m}\|L_j\|^2.
\end{align}

The query complexity and gate complexity for simulating $d$-sparse Lindbladians 
$\mathcal{L}$ are
\begin{align}
O\bigl(\tau\, \mathrm{polylog}(mq\tau/\epsilon)\mathrm{poly}(d,n)\bigr), 
\end{align}
where $\tau = t\|\mathcal{L}\|_{\mathsf{ops}}$.
We sketch the analysis in section~\ref{sec:sparse}.
\item
We expect some of the methodologies in~\cite{BCCKS15,BCK15,LowChuang16,PP16} to be adaptable to the Lindblad evolution simulation problem (in conjunction with our variant of the LCU construction and oblivious amplitude amplification), but have not investigated this.
\end{enumerate}

\section{Brief summary of novel techniques}\label{sec:novel}

As noted in subsection~\ref{sec:previous}, for the case of Hamiltonian evolution, a series of recent quantum algorithms whose scaling is $O(t\,\polylog(t/\epsilon))$ has been discovered which improve on what has been accomplished using the longstanding Trotter-Suzuki decomposition.
One of the main tools that these algorithms employ is a remarkable circuit construction that is based on a certain decomposition of \emph{unitary} operations (or \emph{near-unitary} operations) into a linear combination of unitaries.
We refer to this construction as the \textit{standard LCU method}.

For the case of Lindblad evolution, the operations that arise are \emph{channels} that are not generally unitary.
Some channels are \emph{mixed unitary}, which means that they can be expressed as a randomly chosen unitary (say with probabilities $p_0,\dots,p_{m-1}$ on the unitaries $U_0,\dots,U_{m-1}$).
For such channels, the standard LCU method can be adapted along the lines of first randomly sampling 
$j \in \{0,\dots,m-1\}$ and then applying the standard LCU method to the unitary $U_j$.
However, there exist channels that are \emph{not} mixed unitary---and such channels can arise from the Lindblad equation.
A different reductionist approach is to express these channels in the Stinespring form, as unitary operations that act on a larger system, and then apply the standard LCU method to \emph{those} unitaries; however, as we explain in subsection~\ref{section:standard-lcu}, this approach performs poorly.
We take a different approach that does not involve a reduction to the unitary case: we have developed a new variant of the LCU method that is for channels.
This is explained in section~\ref{sec:new-lcu}. 

Another new technique that we employ is an Oblivious Amplitude Amplification algorithm for isometries (as opposed to unitaries), which is noteworthy because a reductionist approach based on extending isometries to unitaries does not work.
Roughly speaking, this is because our LCU construction turns out to produce an isometry (corresponding to a purification of the channel); however, it does not produce a unitary extension of that isometry.

\subsection{The standard LCU method performs poorly on Stinespring dilations}\label{section:standard-lcu}

Here we show in some technical detail why the standard LCU method performs poorly for Stinespring dilations of channels.
The standard LCU method (explained in detail in subsection~2.1 of~\cite{kothari2014efficient}) for a unitary $V$ expressible as a linear combination of unitaries as
$V = \alpha_0 U_0 + \cdots + \alpha_{m-1}U_{m-1}$
is a circuit construction $W$ that has the property
\begin{align}\label{eq:unitary-lcu}
W \ket{0}\ket{\psi} = \sqrt{p}\ket{0}V\ket{\psi} + \sqrt{1-p}\ket{\Phi^{\perp}}
\end{align}
where $\ket{\Phi^{\perp}}$ has zero amplitude in states with first register $\ket{0}$ (i.e., $(\ketbra{0}{0} \otimes I) \ket{\Phi^{\perp}} = 0$) and
\begin{align}\label{eq:success-standard-LCU}
p = \frac{1}{(\sum_{j=0}^{m-1} \alpha_j)^2}
\end{align}
is the success probability (that arises if the first indicator register is measured).

Consider the \emph{amplitude damping channel}, 
which has two Kraus operators with the following LCU decompositions
\begin{align*}
A_0 
=
\begin{bmatrix}
1 \ & 0 \\
0 & \sqrt{1 - \delta}
\end{bmatrix}
&=
\alpha_{00} \begin{bmatrix}
1 \ & 0 \\
0 & 1
\end{bmatrix}
+
\alpha_{01} \begin{bmatrix}
1 \ & 0 \\
0 & -1
\end{bmatrix} \\
A_1 
=
\begin{bmatrix}
0 & \sqrt{\delta} \\
0 & 0 
\end{bmatrix}
&=
\alpha_{10}
\begin{bmatrix}
0 \ & 1 \\
1 & 0
\end{bmatrix}
+
\alpha_{11}
\begin{bmatrix}
0 \ & 1 \\
-1 & 0
\end{bmatrix},
\end{align*}
where $\alpha_{00} = \frac{1 + \sqrt{1 - \delta}}{2}$, $\alpha_{01} = \frac{1 - \sqrt{1 - \delta}}{2}$, $\alpha_{10} = \frac{\sqrt{\delta}}{2}$, $\alpha_{11} = \frac{\sqrt{\delta}}{2}$.
Evolving an amplitude damping process for time $t$ yields this channel with $\delta = 1 - e^{-t}$.
When $t \ll 1$, $\delta \approx t$, $\alpha_{00} \approx 1 - t/4$, and 
$\alpha_{01} \approx t/4$.

A Stinespring dilation of $V$ and its LCU decomposition can be derived from the above LCU decompositions of $A_0$ and $A_1$ as
\begin{align*}
V =
\begin{bmatrix}
1 & 0 & 0 & 0 \\
0 & \sqrt{1-\delta} & -\sqrt{\delta} & 0  \\
0 & \sqrt{\delta} & \sqrt{1-\delta} & 0 \\
0 & 0 & 0 & 1 
\end{bmatrix}
&=
\alpha_{00}\!
\begin{bmatrix}
1 & 0 & 0 & 0 \\
0 & 1 & 0 & 0  \\
0 & 0 & 1 & 0 \\
0 & 0 & 0 & 1 
\end{bmatrix} +
\alpha_{01}\!
\begin{bmatrix}
1 & 0 & 0 & 0 \\
0 & -1 & 0 & 0  \\
0 & 0 & -1 & 0 \\
0 & 0 & 0 & 1 
\end{bmatrix} \\
&+
\alpha_{10}\!
\begin{bmatrix}
0 & 0 & 0 & -1 \\
0 & 0 & -1 & 0  \\
0 & 1 & 0 & 0 \\
1 & 0 & 0 & 0
\end{bmatrix} +
\alpha_{11}\!
\begin{bmatrix}
0 & 0 & 0 & 1 \\
0 & 0 & -1 & 0  \\
0 & 1 & 0 & 0 \\
-1 & 0 & 0 & 0
\end{bmatrix}.
\end{align*}
Applying the standard LCU method here results in a success probability (computed from Eq.~\eqref{eq:success-standard-LCU}) of
\begin{align*}
\frac{1}{\bigl(\alpha_{00}+\alpha_{01}+\alpha_{10}+\alpha_{11}\bigr)^2} =
\frac{1}{\bigl(1 + \sqrt{\delta}\bigr)^2} = 1 - 2\sqrt{\delta} + \Theta(\delta).
\end{align*}
For small time evolution $t$, the failure probability is $\Theta(\sqrt t)$, which is prohibitively expensive.
It means that the process can be repeated at most $\Theta(1/\sqrt{t})$ times until the cumulative failure probability becomes a constant.
The amount of evolution time (of the amplitude damping process) that this corresponds to is
\begin{align*}
\Theta\Bigl(\frac{1}{\sqrt{t}}\Bigr) \cdot t = \Theta(\sqrt{t}),
\end{align*}
which is subconstant as $t \rightarrow 0$.
This creates a problem in the general Lindblad simulation.

Our new LCU method for channels (explained in section~\ref{sec:new-lcu}) achieves the higher success probability
\begin{align*}
\frac{1}{\bigl(\alpha_{00}+\alpha_{01}\bigr)^2+\bigl(\alpha_{10}+\alpha_{11}\bigr)^2} =
\frac{1}{1 + \delta} = 1 - \delta + \Theta(\delta^2).
\end{align*}
For small time evolution $t$, the failure probability is $\Theta(t)$.
Now, the process can be repeated $\Theta(1/t)$ times until the cumulative failure probability becomes a constant, which corresponds to evolution time
\begin{align*}
\Theta\Bigl(\frac{1}{t}\Bigr) \cdot t = \Theta(1),
\end{align*}
which is constant as $t \rightarrow 0$.
Since this is consistent with what arises in the algorithm of simulating Hamiltonian evolution in~\cite{BCCKS14,BCCKS15}, the methodologies used therein, with various adjustments, can be used to obtain the simulation bounds.

\section{New LCU method for channels and completely positive maps}
\label{sec:new-lcu}
Let $A_0, \dots, A_{m-1}$, linear operators on $\mathbb{C}^{2^n}$
($n$-qubit states), be the Kraus operators of a channel. Suppose that, for
each $j \in \{0,\dots,m-1\}$, we have a decomposition of $A_j$ as a linear
combination of unitaries in the form 
\begin{align}\label{lcu-channel}
A_j = \sum_{k=0}^{q-1} \alpha_{jk} U_{jk},
\end{align}
where, for each $j \in \{0,\dots,m-1\}$ and $k \in \{0,\dots,q-1\}$, $\alpha_{jk} \ge 0$ and $U_{jk}$ is unitary.

The objective is to implement the channel in terms of the implementations of $U_{jk}$'s.
We will describe a circuit $W$ and fixed state $\ket{\mu}$ such that, for any $n$-qubit state $\ket{\psi}$,
\begin{align}\label{eq:property-of-W}
	W\ket{0}\ket{\mu}\ket{\psi} =
	\sqrt{p}\ket{0}\left(\sum_{j=0}^{m-1}\ket{j}A_j\ket{\psi}\right) +
	\sqrt{1-p}\ket{\Phi^{\bot}},
\end{align}
where $(\ketbra{0}{0} \otimes I \otimes I) \ket{\Phi^{\bot}} = 0$ and 
\begin{align}\label{old-lcu}
p = \frac{1}{\sum_{j=0}^{m-1} (\sum_{k=0}^{q-1} \alpha_{jk})^2}\ 
\end{align}
is called the success probability parameter (which is realized if the first register is measured).
Note that the isometry $\ket{\psi} \mapsto \sum_{j=0}^{m-1}\ket{j}A_j\ket{\psi}$ is the channel in purified form.

The circuit $W$ is in terms of two gates.
One gate is a \emph{multiplexed-$U$ gate}, denoted by $\mbox{multi-}U$ such that, for all $j \in \{0,\dots,m-1\}$ and $k \in \{0,\dots,q-1\}$, 
\begin{align}\label{eq:mU}
\mbox{multi-}U\ket{k}\ket{j}\ket{\psi} = \ket{k}\ket{j}U_{jk}\ket{\psi}.
\end{align}
The other gate is a \emph{multiplexed-$B$ gate}, denoted by $\mbox{multi-}B$, such that, for all $j \in \{0,\dots,m-1\}$,
\begin{align}\label{eq:mB}
\mbox{multi-}B\ket{0}\ket{j} = \left(\frac{1}{\sqrt{s_j}} \sum_{k=0}^{q-1}
\sqrt{\alpha_{jk}}\ket{k}\right)\ket{j},
\end{align}
where 
\begin{align}\label{eq:s_j}
s_j = \sum_{k=0}^{q-1} \alpha_{jk}.
\end{align}
Define the state $\ket{\mu}$ (in terms of $s_0,\dots,s_{m-1}$ from Eq.~\eqref{eq:s_j})
\begin{align}\label{eq:mu}
\ket{\mu} = \frac{1}{\sqrt{\sum_{j=1}^{m-1} s_j^2}} \sum_{j=0}^{m-1} s_j \ket{j}.
\end{align}
Define the circuit $W$ (acting on $\mathbb{C}^{q}\otimes\mathbb{C}^{m}\otimes\mathbb{C}^{2^n}$) as  
\begin{align}\label{W}
  W = (\mbox{multi-}B^{\dag}\otimes I)\mbox{multi-}U(\mbox{multi-}B\otimes I).
\end{align}
The LCU construction with the circuit $W$ with its initial state 
$\ket{0}\otimes\ket{\mu}\otimes\ket{\psi}$ is illustrated in Fig.~\ref{gadget}. 
\begin{figure}[!ht]
	\centering
	\includegraphics[width=.4\textwidth]{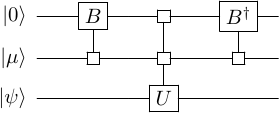}
	\caption{\small The circuit $W$ for simulating a channel using the new LCU method.\label{gadget}}
\end{figure}
In this figure, we refer to the first register as the \emph{indicator register} (as it indicates whether the computation succeeds at the end of this operation), the second register as the \emph{purifier register} (as it is used to purify the channel when the computation succeeds), and the third register as the \emph{system register} (as it contains the state being evolved).

In the following lemma, Eq.~\eqref{eq:property-of-W} is shown to apply where 
$A_0,\dots,A_{m-1}$ are arbitrary linear operators (i.e., Kraus operators of a completely positive map that is not necessarily trace-preserving).
If the map is also trace-preserving then 
$\sum_{j=0}^{m-1}\ket{j}A_j\ket{\psi}$ and $\ket{\Phi^{\perp}}$ are normalized states and the success probability parameter $p$ is the actual success probability realized if the first register is measured; otherwise, these need not be the case.
In subsequent sections, we will apply this lemma in a context where the trace-preserving condition is \emph{approximately} satisfied.
\begin{lemma}\label{lemma:lcu}
Let $A_0, \ldots, A_{m-1}$ be the Kraus operators of a completely positive map. Suppose that each $A_j$ can be written in the form of Eq.~\eqref{lcu-channel}. Let $\mbox{multi-}U$, $\mbox{multi-}B$, $W$, and $\ket{\mu}$ be defined as above.
Then applying the unitary operator $W$ on any state of the form $\ket{0}\ket{\mu}\ket{\psi}$ produces the state 
	\begin{align*}
	\sqrt{p}\ket{0}\left(\sum_{j=0}^{m-1} \ket{j}A_j\ket{\psi}\right) +
	\sqrt{1-p}\ket{\Phi^{\bot}},
	\end{align*}
	where $(\ketbra{0}{0} \otimes I \otimes I) \ket{\Phi^{\bot}} = 0$, and
	\begin{align*}
	p = \frac{1}{\sum_{j=0}^{m-1} \left(\sum_{k=0}^{q-1}\alpha_{jk}\right)^2}\ .
	\end{align*}
\end{lemma}

\begin{proof}
	First consider the state $\ket{0}\ket{j}\ket{\psi}$ for any
	$j\in\{0,\ldots,m-1\}$. Applying $W$ on this state is the standard LCU method \cite{kothari2014efficient}:
	\begin{align}
	  W\ket{0}\ket{j}\ket{\psi} =& (\mbox{multi-}B^{\dag}\otimes
	I)\mbox{multi-}U(\mbox{multi-}B\otimes
	I)\ket{0}\ket{j}\ket{\psi} \\
	=& \frac{1}{\sqrt{s_j}}(\mbox{multi-}B^{\dag}\otimes
	I)\mbox{multi-}U\left(\sum_{k=0}^{q-1}\sqrt{\alpha_{jk}}\ket{k}\right)\ket{j}\ket{\psi}
	\\
	=& \frac{1}{\sqrt{s_j}}(\mbox{multi-}B^{\dag}\otimes I)
	\left(\sum_{k=0}^{q-1}\sqrt{\alpha_{jk}}\ket{k}\ket{j}U_{jk}\ket{\psi}\right)
	\\
	=&
	\frac{1}{s_j}\ket{0}\ket{j}\left(\sum_{k=0}^{q-1}\alpha_{jk}U_{jk}\ket{\psi}\right)
	+ \sqrt{\gamma_j}\ket{\Phi_j^{\bot}} \\
	=& \frac{1}{s_j}\ket{0}\ket{j}A_j\ket{\psi}
	+ \sqrt{\gamma_j}\ket{\Phi_j^{\bot}},
	\end{align}
	where $\ket{\Phi_j^{\bot}}$ is a state satisfying $(\ketbra{0}{0}\otimes
	I\otimes I)\ket{\Phi_j^{\bot}} = 0$ and $\gamma_j$ is some normalization
	factor.
	
	Up to this point, if the indicator register were measured and $\ket{0}$ were observed as the ``success'' case as in the standard LCU method, then the state of the purifier and the system register collapses to $\ket{j}A_j\ket{\psi}$. However, this is not a meaningful quantum sate, as it only captures one Kraus operator of a quantum map. Now we use this specially designed quantum state $\ket{\mu}$ to obtain the desired purification state.
	We use the superposition $\ket{\mu}$ instead of $\ket{j}$ in the
	second register then, by linearity, we have
	\begin{align}
	W\ket{0}\ket{\mu}\ket{\psi} =
	\sqrt{p}\ket{0}\left(\sum_{j=0}^{m-1}\ket{j}A_j\ket{\psi}\right) +
	\sqrt{1-p}\ket{\Phi^{\bot}},
	\end{align}
	where $(\ketbra{0}{0}\otimes I\otimes I)\ket{\Phi^{\bot}} = 0$ and 
	$p = \frac{1}{\sum_{j=0}^{m-1} s_j^2}$.
\end{proof}

\section{Overview of the main result, Theorem~\ref{thm:main_thm}}\label{sec:main_thm}

In this section we show how to apply our new LCU method in order to prove our main result, Theorem~\ref{thm:main_thm}. 

We begin by reviewing some basic notation and definitions pertaining to superoperators.
For a finite-dimensional Hilbert space $\mathcal{X}$, let $L(\mathcal{X})$ denote the set set of all linear operators mapping $\mathcal{X}$ to $\mathcal{X}$.
We consider mappings that are linear operators from $L(\mathcal{X})$ to $L(\mathcal{X})$ (sometimes referred to as superoperators).
For such a mapping $\mathcal{T}$, we write $\mathcal{T}[X]$ to denote the result of $\mathcal{T}$ applied to $X \in L(\mathcal{X})$ (the square brackets $[\cdot]$ around the argument are to highlight that the mapping acts on linear operators).
Also, we use $\mathbbm{1}_{\mathcal{X}}$ to denote the identity map on $L(\mathcal{X})$.

For a linear mapping on $\mathcal{T} : L(\mathcal{X}) \rightarrow L(\mathcal{X})$, the \emph{induced trace norm}%
\footnote{Some authors refer to the induced trace norm as the 1$\to$1 norm} of $\mathcal{T}$ \cite{Watrous-notes}
is defined as
\begin{align}
  \pnorm{1}{\mathcal{T}} = \max_{\pnorm{1}{Q}=1}\pnorm{1}{\mathcal{T}[Q]},
\end{align}
where $\|Q\|_1$ denotes the trace norm of $Q \in L(\mathcal{X})$.
The \emph{diamond norm} of $\mathcal{T}$ is defined as
\begin{align}
  \pnorm{\diamond}{\mathcal{T}} = \pnorm{1}{\mathcal{T}\otimes\mathbbm{1}_{\mathcal{X}}}.
\end{align}

Next, we show that, for Lindbladians specified by Eqns.~\eqref{eq:lindbladian}, \eqref{eq:H-lcu} and \eqref{eq:jump-lcu}, the quantities $\|\mathcal{L}\|_{\textsf{pauli}}$ (defined in Eq.~\eqref{cost-lindbladian})
and 
$\|\mathcal{L}\|_{\textsf{ops}}$ (defined in Eq.~\eqref{eq:sparse-norm})
satisfy
\begin{align}\label{eq:norm-inequalities}
\|\mathcal{L}\|_{\diamond} \le 2\|\mathcal{L}\|_{\textsf{ops}} \le 2\|\mathcal{L}\|_{\textsf{pauli}}.
\end{align}
For the first inequality in Eq.~\eqref{eq:norm-inequalities}, 
note that $\|\mathcal{L}\|_1 \le 2\|\mathcal{L}\|_{\textsf{ops}}$ 
holds by the triangle inequality and the fact that, for any $X \in L(\mathcal{X})$ such that $\pnorm{1}{X} = 1$,
\begin{align}
\|[H,X]\|_1 &\le 2\|H\| \\
\bigl\|L_j X L_j^{\dag}\bigr\|_1 &\leq
    \bigl\|L_j\bigr\| \bigl\| X \bigr\|_1 \bigl\|L_j^{\dag}\bigl\|=\bigl\|L_j\bigr\|^2.
\end{align}
Then, since $\|M \otimes I\| = \|M\|$ for any linear operator $M$,  the first inequality in Eq.~\eqref{eq:norm-inequalities} follows.
The second inequality in Eq.~\eqref{eq:norm-inequalities} follows from the fact that, if $H$ and $L_0,\dots,L_{m-1}$ are specified as in Eqns.~\eqref{eq:H-lcu} and \eqref{eq:jump-lcu}, then 
\begin{align}
\|H\| \le \sum_{k=0}^{q-1}\beta_{0k} \ \ \ \ \mbox{and } \quad 
\|L_j\| \le \sum_{k=0}^{q-1}\beta_{jk}, \quad \text{ for all } j \in \{1, \ldots, m\}.
\end{align}

We are now ready to present the details of our construction for the proof of Theorem~\ref{thm:main_thm}. The overall structure is similar to that in~\cite{BCCKS14} and~\cite{BCCKS15}, with the main novel ingredient being our variant of the LCU construction (explained in section~\ref{sec:new-lcu}) and also a variant of oblivious amplitude amplification for isometries. For clarity, the details are organized into the following subsections, whose content is summarized as:
\begin{enumerate}
\item
  In subsection~\ref{section:appro-map}, we describe a simple mapping $\mathcal{M}_{\delta}$ in terms of Kraus operators that are based on the operators in $\mathcal{L}$.
For small $\delta$, $\mathcal{M}_{\delta}$ is a good approximation of $e^{\mathcal{L}\delta}$.
\item
  In subsection~\ref{section:simulate-maps}, we show how to simulate the mapping $\mathcal{M}_{\delta}$ in the sense of Lemma~\ref{lemma:lcu}, with success probability parameter $1-O(\delta)$.
\item
  In subsection~\ref{section:segments}, we show how to combine $r$ simulations of 
$\mathcal{M}_{O(1/r)}$ so as to obtain cumulative success probability parameter $1/4$.
Conditional on success, this produces a good approximation of constant-time Lindblad evolution. 
\item
  In subsection~\ref{section:OAA}, we show how to apply a modified version of oblivious amplitude amplification to unconditionally simulate an approximation of constant-time Lindblad evolution.
\item
  In subsection~\ref{section:hamming}, we show how to reduce the number of multiplexed Pauli gates by a concentration bound on the amplitudes associated with nontrivial Pauli gates.
\item
  In subsection~\ref{section:gate-count}, we bound the total number of gates and combine the simulations for segments in order to complete the proof of Theorem~\ref{thm:main_thm}. 
\end{enumerate}

\subsection{A simple mapping $\mathcal{M}_{\delta}$ that approximates $e^{\mathcal{L}\delta}$ for small $\delta$}
\label{section:appro-map}
Here, we show how to approximate Lindblad evolution for small time $\delta$, namely $e^{\mathcal{L}\delta}$, by a mapping $\mathcal{M}_{\delta}$ that can be described in terms of $m+1$ Kraus operators, where the precision of the approximation is $O(\delta^2)$.

Following an approach described in~\cite{Preskill-notes}, define
the quantum map $\mathcal{M}_{\delta}$ as 
\begin{align}
\mathcal{M}_{\delta} [Q] = \sum_{j=0}^{m} A_j Q A_j^{\dag},
\end{align}
where
\begin{align}\label{eq:instantanious-channel}
A_0 = I-\frac{\delta}{2}\sum_{j=1}^mL_j^{\dag}L_j-i\delta H \ \ \ \ 
\mbox{and, for $j \in \{1,\dots,m\}$,} \ \ \ \ 
A_j = \sqrt{\delta}L_j.
\end{align}

Note that, in general, $\mathcal{M}_{\delta}$ does not satisfy the trace-preserving condition for a quantum channel; however, it satisfies an approximate version of it:
\begin{align}
  \norm{\sum_{j=0}^mA_j^{\dag}A_j - I} 
  &= \norm{\frac{\delta^2}{4}\Biggl(\,\sum_{j=1}^mL_j^{\dag}L_j\Biggr)^2+\delta^2H^2} \\
  &\leq \delta^2 \norm{\Biggl(\,\sum_{j=1}^mL_j^{\dag}L_j + H\Biggr)^2} \\
  &\leq \delta^2\Biggl(\,\sum_{j=1}^m\norm{L_j^{\dag}L_j}+\norm{H}\Biggr)^2 \\
\label{eq:appro-krauss}
&= (\delta\pnorm{\textsf{ops}}{\mathcal{L}})^2.
\end{align}

We now show that 
\begin{align}\label{eq:Mvs.e^L}
\|\mathcal{M}_{\delta} - e^{\mathcal{L}\delta}\|_{\diamond} \le 
5(\delta \pnorm{\textsf{ops}}{\mathcal{L}})^2.
\end{align}
To do this, we introduce an intermediate quantum map, $\mathbbm{1}+\delta\mathcal{L}$ (mapping $\rho$ to $\rho+\delta\mathcal{L}[\rho]$), and show that 
\begin{align}\label{eq:Mvs.1+L}
\pnorm{\diamond}{\mathcal{M}_{\delta} - (\mathbbm{1} + \delta\mathcal{L})} \le (\delta\pnorm{\textsf{ops}}{\mathcal{L}})^2
\end{align}
and then Eq.~\eqref{eq:Mvs.e^L} follows from the fact that 
\begin{align}
\bigl\|(\mathbbm{1} + \delta\mathcal{L}) - e^{\mathcal{L}\delta}\bigr\|_{\diamond}
&\le (\delta\pnorm{\diamond}{\mathcal{L}})^2 \label{eq:1+Lvs.e^L} \\
&\le (2\delta\pnorm{\textsf{ops}}{\mathcal{L}})^2.
\end{align}
For completeness, Eq.~\eqref{eq:1+Lvs.e^L} is proven in Appendix~\ref{app:1plusL}.
In order to show Eq.~\eqref{eq:Mvs.1+L}, note that 
for any operator $Q$ on $\mathcal{H}\otimes\mathcal{K}$ with $\pnorm{1}{Q}=1$, \small
\begin{align}\label{eq:dist-intermediate-2}
\Bigl\|\bigl(\mathcal{M}_{\delta}\otimes\mathbbm{1}_{\mathcal{K}} - (\mathbbm{1}_{\mathcal{H}} &+ \delta\mathcal{L})\otimes\mathbbm{1}_{\mathcal{K}}\bigr) [Q]\Bigr\|_{1} \\
& =\Biggl\|\sum_{j=0}^m(A_j\otimes I)Q (A_j\otimes I)^{\dag} -
  (Q+\delta(\mathcal{L}\otimes\mathbbm{1}_{\mathcal{K}})[Q])\Biggr\|_1 \\
& =
\Biggl\|
\frac{\delta^2}{4}
\Bigl(\sum_{j=1}^m L_j^{\dag}L_j\otimes I\Bigr)
Q 
\Bigl(\sum_{j'=1}^m L_{j'}^{\dag}L_{j'}\otimes I\Bigr) 
- \frac{\delta^2}{2}i\sum_{j=1}^m(L_j^{\dag}L_j\otimes I)Q (H\otimes I)
\Biggr. \\
& \ \ \ \ \ \ \ \Biggl.
+\frac{\delta^2}{2}i(H\otimes I)Q\sum_{j=1}^m(L_j^{\dag}L_j\otimes I)+\delta^2(H\otimes I)Q (H\otimes I) \Biggr\|_1 \\
& \leq \delta^2\Biggl(
\Bigl\|\sum_{j=1}^m L_j^{\dag}L_j\otimes I\Bigr\|^2 + 2\Bigl\|H\otimes I\Bigr\|\Bigl\|\sum_{j=1}^m L_j^{\dag}L_j\otimes I\Bigr\|+\Bigl\|H\otimes I\Bigr\|^2\Biggr) \\
& \leq \delta^2\Biggl(
\Bigl\|\sum_{j=1}^m L_j^{\dag}L_j\otimes I\Bigr\|+\Bigl\|H\otimes I\Bigr\|\Biggr)^2 \\
& \leq \delta^2\left(\sum_{j=1}^m\norm{L_j}^2+\norm{H}\right)^2 \\
& \leq \delta^2\pnorm{\textsf{ops}}{\mathcal{L}}^2.
\end{align}
\normalsize
This completes the proof of Eq.~\eqref{eq:Mvs.e^L}.

\subsection{Approximating $\mathcal{M}_{\delta}$ by a quantum circuit via the new LCU method}
\label{section:simulate-maps}
Here we show how to construct a quantum circuit that computes an approximation of $\mathcal{M}_{\delta}$ along the lines of Eq.~\eqref{eq:property-of-W} using the new LCU method.

By substituting Eqns.~\eqref{eq:H-lcu} and \eqref{eq:jump-lcu} into Eq.~\eqref{eq:instantanious-channel}, we have
\begin{align}\label{eq:A-lcu-1}
  A_j &= \sqrt{\delta}\sum_{k=0}^{q-1}\beta_{jk}V_{jk} \text{, \quad for $j\in\{1,\ldots,m\}$, and} \\
  A_0 &= I-\frac{\delta}{2}\sum_{j=1}^m\left(\sum_{k=0}^{q-1}\beta_{jk}V_{jk}\right)^{\dag}\left(\sum_{l=0}^{q-1}\beta_{jl}V_{jl}\right) - i\delta\sum_{k=0}^{q-1}\beta_{0k}V_{0k} \\
  \label{eq:A-lcu-2} 
  &= I + \frac{\delta}{2}\sum_{j=1}^m\sum_{k=0}^{q-1}\sum_{l=0}^{q-1}\beta_{jk}\beta_{jl}\left(-V_{jk}^{\dag}V_{jl}\right) + \delta\sum_{k=0}^{q-1}\beta_{0k}\left(-iV_{0k}\right).
\end{align}
Note that Eqns.~\eqref{eq:A-lcu-1} and \eqref{eq:A-lcu-2} are expressing the Kraus operators $A_0, \dots, A_m$ as tensor products of Paulis (i.e., they are of the form of Eqns.~\eqref{eq:H-lcu} and~\eqref{eq:jump-lcu}).
Therefore, by Lemma~\ref{lemma:lcu}, the circuit construction of $W$ in Fig.~\ref{gadget} and the state $\ket{\mu}$ from Eq.~\eqref{eq:mu} satisfy the following property.
For any state $\ket{\psi}$, 
\begin{align}
W\ket{0}\ket{\mu}\ket{\psi} 
=  \sqrt{p}\ket{0}\left(\sum_{j=0}^m\ket{j}A_j\ket{\psi}\right)+\sqrt{1-p}\ket{\Phi^{\bot}},
\end{align}
where $\ket{\Phi^{\bot}}$ satisfies 
$(\ketbra{0}{0}\otimes I\otimes I)\ket{\Phi^{\bot}} = 0$ and 
\begin{align}
p =\frac{1}{\sum_{j=0}^ms_j^2},
\end{align}
where 
\begin{align}\label{eq:sj}
  s_j &= \sqrt{\delta}\sum_{k=0}^{q-1}\beta_{jk}, \text{\quad for $j\in\{1,\ldots,m\}$, \ \ and} \\
  \label{eq:s0}
  s_0 &= 1+\frac{\delta}{2}\sum_{j=1}^m\sum_{k=0}^{q-1}\sum_{l=0}^{q-1}\beta_{jk}\beta_{jl} + \delta\sum_{k=0}^{q-1}\beta_{0k}.
\end{align}
(The values of $s_0,\dots,s_m$ are directly from Eqns.~\eqref{eq:A-lcu-1} and~\eqref{eq:A-lcu-2}.)

To simplify the expression for the success probability parameter, it is convenient to define 
\begin{align}\label{eq:c_j}
c_j = \sum_{k=0}^{q-1}\beta_{jk},
\end{align}
for $j\in\{0,\ldots,m\}$.
Then we can rewrite Eqns.~\eqref{eq:sj} and \eqref{eq:s0} as
\begin{align}\label{eq:sj-c}
  s_j &= \sqrt{\delta}c_j, \text{\quad for $j\in\{1,\ldots,m\}$, \ \ \ and} \\
  s_0 &= 1+\frac{\delta}{2}\sum_{j=1}^mc_j^2+\delta c_0 \label{eq:s0-c}
\end{align}
and
\begin{align}
  p & = \frac{1}{\sum_{j=0}^ms_j^2}  \\
  & =
  \frac{1}{\left(1+\frac{\delta}{2}\sum_{j=1}^mc_j^2+\delta c_0\right)^2+\sum_{j=1}^mc_j^2\delta} \\
& =
  \frac{1}{1+2\delta\sum_{j=1}^mc_j^2+2\delta c_0+ \delta^2\left(\frac{1}{2}\sum_{j=1}^mc_j^2+ c_0\right)^2}\\
& = \frac{1}{1 + 2\delta\Bigl(\sum_{j=1}^mc_j^2+c_0\Bigr) 
+ \frac{\delta^2}{4}\Bigl(\sum_{j=1}^mc_j^2 + 2c_0\Bigr)^2} \\
& = \frac{1}{1 + 2\delta\pnorm{\textsf{pauli}}{\mathcal{L}} 
+ \frac{\delta^2}{4}\bigl(\pnorm{\textsf{pauli}}{\mathcal{L}} + c_0\bigr)^2} \label{eq:probability-parameter} \\[1.5mm]
& = 1-2 \delta \pnorm{\textsf{pauli}}{\mathcal{L}} 
- O\bigl(\delta^2\pnorm{\textsf{pauli}}{\mathcal{L}}^2\bigr).
\end{align}

Note that, since $\mathcal{M}_{\delta}$ is only an approximate channel, the success probability parameter $p$ does not correspond to the actual probability of outcome $0$ if the indicator register is measured; 
however, it can be shown that $p$ is within $O(\delta^2)$ of the actual success probability.
We do not show this here; our analysis will be in terms of the cumulative error arising in circuit constructions in the subsequent sections (which consist of several instances of the construction from this section).

\subsection{Simulating $r$ iterations of $\mathcal{M}_{O(1/r)}$ with constant success probability}
\label{section:segments}
\label{section:constant-prob}
In this section we iterate the construction from the previous section $r$ times, with $\delta = O(1/r)$.

The resulting success probability parameter associated with $\mathcal{M}_{\delta}^{(r)}$ is $p^r = \bigl(1 - O(1/r)\bigr)^r$, which converges to a constant.
We can tune the parameter $\delta$ so that $p^r = 1/4$ holds exactly.
This is accomplished by setting $p = 4^{-1/r}$ and then solving for $\delta$ in Eq.~\eqref{eq:probability-parameter}, yielding the positive solution
\begin{align}\label{eq:setting-of-delta}
\delta &=
\frac{-\pnorm{\textsf{pauli}}{\mathcal{L}} + \sqrt{\pnorm{\textsf{pauli}}{\mathcal{L}}^2+\frac{1}{4}\bigl(\pnorm{\textsf{pauli}}{\mathcal{L}}+c_0\bigr)^2\bigl(4^{1/r} - 1\bigr)}}{\frac{1}{4}\bigl(\pnorm{\textsf{pauli}}{\mathcal{L}}+c_0\bigr)^2}\\[1.5mm]
&= \Biggl(\frac{\ln(2)}{\pnorm{\textsf{pauli}}{\mathcal{L}}}\Biggr)\frac{1}{r} + O\Bigl(\frac{1}{r^2}\Bigr).
\end{align}
 
The circuit that implements $\mathcal{M}_{\delta}^{(r)}$ uses an initial state $(\ket{0}\ket{\mu})^{\otimes r}\ket{\psi}$, which can be reordered to $\ket{0}^{\otimes r}\ket{\mu}^{\otimes r}\ket{\psi}$.
It consists of $r$ instances of $W$, each with \emph{separate} indicator and purifier registers, but with \emph{the same} system register.
Let
$\widehat{W}$ denote this unitary operator (consisting of $r$ applications of $W$ on different indicator and purifier registers).
For each 
$\jj =j_0\ldots j_{r-1} \in\{0,\ldots,m\}^{r}$, define $\widehat{A}_{\jj}$ as
\begin{align}
\widehat{A}_{\jj} =
A_{j_0}\cdots A_{j_{r-1}}.
\end{align}
We can conclude that
\begin{align}
  \widehat{W}\Bigl(\ket{0}^{\otimes r}\ket{\mu}^{\otimes
    r}\ket{\psi}\Bigr) 
&=
    \sqrt{p^r}\ket{0}^{\otimes r}\Biggl(\,\sum_{\jj\in\{0,\ldots,m\}^{r}}\ket{\jj}\widehat{A}_{\jj}\ket{\psi}\Biggr)
    + \sqrt{1-p^r}\ket{\widehat{\Phi}^{\bot}} \\
&= \frac{1}{2}\ket{0}^{\otimes
r}\Biggl(\,\sum_{\jj\in\{0,\ldots,m\}^{r}}\ket{\jj}\widehat{A}_{\jj}\ket{\psi}\Biggr)
+ \frac{\sqrt{3}}{2}\ket{\widehat{\Phi}^{\bot}},\label{eq:D}
\end{align}
where $\ket{\widehat{\Phi}^{\bot}}$ satisfies
$(\ketbra{0}{0}^{\otimes r}\otimes I^{\otimes r}\otimes
I)\ket{\widehat{\Phi}^{\bot}}=0$.

Note that this conditionally simulates $\mathcal{M}_{\delta}^{(r)}$, and 
$\mathcal{M}_{\delta}^{(r)}$ approximates $e^{\mathcal{L}t}$, for 
\begin{align}\label{eq:def-t}
t &= r \delta \\
&= \frac{\ln(2)}{\pnorm{\textsf{pauli}}{\mathcal{L}}} + O\Bigl(\frac{1}{r}\Bigr).
\end{align}
The approximation is in the sense
\begin{align}\label{eq:dia-dist-1}
  \bigl\|\mathcal{M}_{\delta}^{(r)}-e^{\mathcal{L}t}\bigr\|_{\diamond} = O\Bigl(\frac{1}{r}\Bigr).
\end{align}

If the desired evolution time is such that $t\pnorm{\textsf{pauli}}{\mathcal{L}} < \ln(2)$ then the success probability parameter resulting from this approach is larger than $1/4$; however, it can be diluted to be exactly $1/4$ using a method described in \cite{BCCKS14} that employs an additional qubit as part of the indicator register.

Next we show how to use oblivious amplitude amplification to achieve
perfect success probability.

\subsection{Oblivious amplitude amplification for isometries}
\label{section:OAA}
There are two hurdles for applying oblivious amplitude amplification in our
construction. First, the purified quantum state corresponding to the
success case is not a normalized quantum state, as the Kraus operators of
$\mathcal{M}_{\delta}^{(r)}$ do not satisfy the trace-preserving
condition. Second, the operation corresponding to the success case is an
isometry (rather than a unitary), because part of the registers in the initial
state is restricted to $(\ket{\mu_0}\cdots\ket{\mu_{m-1}})^{\otimes r}$.

The second hurdle is resolved by using different projectors in the
amplitude amplification operator. For the first hurdle, we show that it
only causes a small error. To begin with, we examine how far it is for the
Kraus operators to satisfy the trace-preserving condition, and this
quantity will be used later in the proof. By repeatedly applying Eq.~\eqref{eq:appro-krauss}, we have
\begin{align}
\Biggl\|\,\sum_{\jj\in\{0,\ldots,m\}^r}\widehat{A}_{\jj}^{\dag}\widehat{A}_{\jj} -I \,\Biggr\| 
&=
\Biggl\|\,\sum_{j_0\cdots j_{r-1}\in\{0,\ldots,m\}^r}(A_{j_{r-1}}^{\dag}\cdots A_{\jj_0}^{\dag})
	(A_{j_0}\cdots A_{j_{r-1}}) - I\,\Biggr\| \\[1mm]
&\le 
r\bigl(\delta\pnorm{\textsf{pauli}}{\mathcal{L}}\bigr)^2 \\[1.5mm]
\label{eq:seg-appro-krauss}
&= (\ln(2))^2/r + O(1/r^2),
\end{align}
where the second equality follows from substituting the value of $\delta$ from Eq.~\eqref{eq:setting-of-delta}.

Before we present the oblivious amplitude amplification construction, we 
introduce more notations for convenience.
For any $\ket{\psi}$, let $\ket{\Psi}$ denote the initial state
\begin{align}
\ket{\Psi}:=\ket{\widehat{0}}\ket{\widehat{\mu}}\ket{\psi},
\end{align}
where $\ket{\widehat{0}} = \ket{0}^{\otimes r}$, and $\ket{\widehat{\mu}}
= \ket{\mu}^{\otimes r}$.
Let $\ket{\Phi}$ denote the desired
purification state, i.e.,
\begin{align}
\ket{\Phi} = \ket{\widehat{0}}\left(\sum_{\jj\in\{0,\ldots,m\}^{r}}\ket{\jj}\widehat{A}_{\jj}\ket{\psi}\right).
\end{align}
Let $P_0:=\ketbra{\widehat{0}}{\widehat{0}}\otimes I\otimes I$ and $P_1 :=
\ketbra{\widehat{0}}{\widehat{0}}\otimes\ketbra{\widehat{\mu}}{\widehat{\mu}}\otimes
I$ be two projectors.
By Eq.~\eqref{eq:D}, we have
\begin{align}\label{eq:D1}
  \widehat{W}\ket{\Psi} = \frac{1}{2}\ket{\Phi} +
\frac{\sqrt{3}}{2}\ket{\Phi^{\bot}},
\end{align}
for some state $\ket{\Phi^{\bot}}$ satisfying $P_0\ket{\Phi^{\bot}} = 0$.
Define the unitary operator
\begin{align}
  F =-\widehat{W}(I-2P_1)\widehat{W}^{\dag}(I-2P_0)\widehat{W}
\end{align}
as the oblivious amplitude amplification operator.
We summarize the result in the
following lemma.

\begin{lemma} \label{lemma:OAA}
	For any state $\ket{\psi}$, Let $\ket{\Psi}$, $\ket{\Phi}$, and $F$ be
	defined as above. Then
	\begin{align*}
	\norm{F\ket{\Psi}
		- \ket{\Phi}} = O(1/r).
	\end{align*}
\end{lemma}

To prove this lemma, we need the following lemma, which slightly extends the results of Lemma 2.3 in \cite{kothari2014efficient}.
\begin{lemma}\label{lemma:OAA-2}
	For any $\ket{\psi}$, let $\ket{\Psi}$, $\ket{\Phi}$, $\ket{\Phi^{\bot}}$,
	$P_0$, and $P_1$ be defined as above. Let $\ket{\Psi^{\bot}}$ be a state satisfying the equation
	\begin{align}\label{eq:D2}
	  \widehat{W}\ket{\Psi^{\bot}} =
	\frac{\sqrt{3}}{2}\ket{\Phi} - \frac{1}{2}\ket{\Phi^{\bot}}.
	\end{align}
	Then $P_1\ket{\Psi^{\bot}} = O(1/r)$.
\end{lemma}
\begin{proof}
	Define the operator
	\begin{align}
	Q =\left(\bra{\widehat{0}}\bra{\widehat{\mu}}\otimes
	  I\right)\widehat{W}^{\dag}P_0\widehat{W}\left(\ket{\widehat{0}}\ket{\widehat{\mu}}\otimes
	I\right).
	\end{align}
	For any state $\ket{\psi}$,
	\begin{align}
	\bra{\psi}Q\ket{\psi} =
	\norm{P_0\widehat{W}\left(\ket{\widehat{0}}\ket{\widehat{\mu}}\ket{\psi}\right)}^2 =
	\norm{P_0\left(\frac{1}{2}\ket{\Phi}+\frac{\sqrt{3}}{2}\ket{\Phi^{\bot}}\right)}^2
	= \norm{\frac{1}{2}\ket{\Phi}}^2 = \frac{1}{4} + O(1/r).
	\end{align}
	The last equality holds because $\norm{\ket{\Phi}}^2=1+O(1/r)$, which
	follows from Eq.~\eqref{eq:seg-appro-krauss}.
	Therefore, all the eigenvalues of $Q$ are $\frac{1}{4} + O(1/r)$, and
	we can write
	\begin{align}\label{eq:Q1}
	Q = \frac{1}{4}I+O\Bigl(\frac{1}{r}\Bigr).
	\end{align}
	Now, for any $\ket{\psi}$, we have
	\begin{align}
	Q\ket{\psi} 
	&= \left(\bra{\widehat{0}}\bra{\widehat{\mu}}\otimes
  I\right)\widehat{W}^{\dag}P_0\widehat{W}\left(\ket{\widehat{0}}\ket{\widehat{\mu}}\ket{\psi}\right) \\
  &= \frac{1}{2}\left(\bra{\widehat{0}}\bra{\widehat{\mu}}\otimes I\right)\widehat{W}^{\dag}\ket{\Phi}
	\\
	&= \frac{1}{2}\left(\bra{\widehat{0}}\bra{\widehat{\mu}}\otimes
	I\right)\Bigl(\frac{1}{2}\ket{\Psi}+\frac{\sqrt{3}}{2}\ket{\Psi^{\bot}}\Bigr)
	\\
	\label{eq:Q2}
	&= \frac{1}{4}\ket{\psi} +
	\frac{\sqrt{3}}{4}\left(\bra{\widehat{0}}\bra{\widehat{\mu}}\otimes
	I\right)\ket{\Psi^{\bot}}.
	\end{align}
	The third equality follows from Eqns.~\eqref{eq:D1} and \eqref{eq:D2}. On the
	other hand, by Eq.~\eqref{eq:Q1}, we have
	\begin{align}\label{eq:Q3}
	Q\ket{\psi} = \frac{1}{4}\ket{\psi} + O\Bigl(\frac{1}{r}\Bigr).
	\end{align}
	By Eqns.~\eqref{eq:Q2} and \eqref{eq:Q3}, we have
	\begin{align}
	\left(\bra{\widehat{0}}\bra{\widehat{\mu}}\otimes I\right)\ket{\Psi^{\bot}} =
	O\Bigl(\frac{1}{r}\Bigr),
	\end{align}
	which implies $P_1\ket{\Psi^{\bot}} = O(1/r)$.
\end{proof}

Now we are ready to prove Lemma~\ref{lemma:OAA}. The proof uses the methods in \cite{BCCKS15}.
\begin{proof}[Proof of Lemma~\ref{lemma:OAA}]
  First consider the operator $P_1\widehat{W}^{\dag}P_0\widehat{W}$. We
	have
	\begin{align}
	  P_1\widehat{W}^{\dag}P_0\widehat{W}\ket{\Psi} =
	  \frac{1}{2}P_1\widehat{W}^{\dag}\ket{\Phi} =
	\frac{1}{2}P_1\left(\frac{1}{2}\ket{\Psi} +
	\frac{\sqrt{3}}{2}\ket{\Psi^{\bot}}\right) =
	\frac{1}{4}\ket{\Psi}+O\Bigl(\frac{1}{r}\Bigr),
	\end{align}
    where the second equality follows from Eqns.~\eqref{eq:D1} and~\eqref{eq:D2} and the last equality follows from Lemma~\ref{lemma:OAA-2}.
    Then we
	have
	\begin{align}
	F\ket{\Psi} &= 
	(-\widehat{W}(I-2P_1)\widehat{W}^{\dag}(I-2P_0)\widehat{W}\ket{\Psi} \\
	&=
	  (2P_0\widehat{W}+\widehat{W}-4\widehat{W}P_1\widehat{W}^{\dag}P_0\widehat{W})\ket{\Psi} \\
	&= \ket{\Phi} + O(1/r).
	\end{align}
	Therefore
	$\norm{F\ket{\Psi} - \ket{\Phi}} = O(1/r)$.
\end{proof}

Since $\ket{\Phi}$ is a purification of $\mathcal{M}_{\delta}^{(r)}[\ketbra{\psi}{\psi}]$, Lemma~\ref{lemma:OAA} implies that the circuit of $F$ simulates a Stinespring dilation of $\mathcal{M}_{\delta}^{(r)}$ with error $O(1/r)$. This further implies that
\begin{align}\label{eq:dia-dist-2}
  \bigl\|\mathcal{N}-\mathcal{M}_{\delta}^{(r)}\bigr\|_{\diamond} = O(1/r),
\end{align}
where $\mathcal{N}$ is the quantum channel that $F$ implements by tracing out indicator and purifier registers.

\subsection{Concentration bound and encoding scheme}
\label{section:hamming}
From the previous sections, $r$ is a parameter that determines the precision, which is $O(1/r)$.
Up to this point, to simulate constant-time Lindblad evolution, the number of occurrences of the multiplexed-$U$ gate in our construction is $O(r)$.
In this subsection, we show how to reduce this to 
$O\bigl(\frac{\log(1/\epsilon)}{\log\log(1/\epsilon)}\bigr)$ while only introducing an additional error of $\epsilon$. 

It is important to note that, in light of Eqns.~\eqref{eq:A-lcu-1} and \eqref{eq:A-lcu-2}, there are $O(m)$ Kraus operators for $\mathcal{M}_{\delta}$ and each can be expressed as an LCU of $O(mq^2)$ terms.

Consider the initial state $(\ket{0}\ket{\mu})^{\otimes r}$ of the indicator and purifier registers. The multiplexed-$B$ gates applied on this state are $\mbox{multi-}B^{\otimes r}$.
Note that the first term in Eq.~\eqref{eq:A-lcu-2}
corresponds to the unitary $I$, which need not be performed. The circuit can be rearranged to bypass these operations, as in earlier papers on Hamiltonian evolution (see, for example, \cite{BCCKS14}).

More precisely, we compute the amplitude associated with this $I$ being performed.
For each instance of $W$ acting on $\ket{0}\ket{\mu}\ket{\psi}$, consider the state of indicator and purifier registers which control the multiplexed-$U$ gates (i.e., the state $\mbox{multi-}B\ket{0}\ket{\mu}$).
The state $\ket{0}\ket{0}$ corresponds to unitary $I$.
The amplitude of $\ket{0}\ket{0}$ is
\begin{align}
  \frac{s_{0}}{\sqrt{\sum_{j=0}^ms_j^2}}\frac{1}{\sqrt{s_{0}}}
  &= \sqrt{\frac{s_0}{\sum_{j=0}^ms_{j}^2}} \\[1mm]
&= \sqrt{
\frac{1 + \delta/2\sum_{j=1}^mc_j^2+\delta c_0}{1 + 2\delta\sum_{j=1}^mc_j^2+2\delta c_0 + \Theta(\delta^2(\sum_{j=1}^mc_j^2+c_0)^2)}} \\[2mm]
&= \sqrt{1 -\frac{3\delta}{2}\sum_{j=1}^mc_j^2-\delta c_0 + \Theta\Biggl(\delta^2\Bigl(\sum_{j=1}^mc_j^2+c_0\Bigr)^2\Biggr)},
\end{align}
where $s_j$ and $c_j$ are defined in Eqns.~\eqref{eq:sj}, \eqref{eq:s0}, and~\eqref{eq:c_j} ($j \in \{0,\dots,m\}$), and 
$\delta$ is defined in Eq.~\eqref{eq:setting-of-delta}.

If this indicator and purifier registers are measured in the computational basis then the probability that the outcome is \emph{not} $(0,0)$ is
\begin{align}\label{eq:not0-prob1}
\frac{3\delta}{2}\sum_{j=1}^mc_j^2 +\delta c_0+\Theta\Biggl(\delta^2\Bigl(\sum_{j=1}^mc_j^2+c_0\Bigr)^2\Biggr) &\leq \frac{3}{2}\delta\pnorm{\textsf{pauli}}{\mathcal{L}} + \Theta\bigl(\delta^2\pnorm{\textsf{pauli}}{\mathcal{L}}^2\bigr) \\
&= \frac{3}{2r} + \Theta\Bigl(\frac{1}{r^2}\Bigr).
\end{align}
Therefore, after the $\mbox{multi-}B$ acting on $\ket{0}\ket{\mu}\ket{\psi}$, if the indicator and purifier registers are measured, then the probability that the outcome is not $(0,0)$ is upper-bounded by $\frac{3}{2r}$. 

Roughly speaking, this is qualitatively the same scaling that arises in  Hamiltonian evolution simulation \cite{BCCKS14}, hence the same so-called Hamming weight cut-off applies.
Below is a more precise explanation of this.

  In the indicator and purifier registers, after applying $\mbox{multi-}B$, the computational basis states of the indicator and purifier registers are of the form $\ket{k_0,l_0}\cdots\ket{k_{r-1},l_{r-1}}$. Define the \emph{Hamming weight} of such a state as the number of $i \in \{0,\dots,r-1\}$ such that $(k_i,l_i) \neq (0,0)$. If the indicator and purifier registers are restricted to states that have Hamming weight at most $h$ then the circuit can be restructured so that there are only $h$ occurrences of the multiplexed-$U$ gates.

  Let $X_1, \ldots, X_r$ be $r$ independent random variables with $\Pr[X_j = 1] = \frac{3}{2r}$ and $\Pr[X_j = 0] = 1 - \frac{3}{2r}$ for all $j \in \{1, \ldots, r\}$. Consider the state of the indicator and purifier registers right before multiplexed-$U$ gates are applied (i.e., the state $(\mbox{multi-}B\ket{0}\ket{\mu})^{\otimes r}$). We are interested in how much amplitude is associated with the low Hamming weight states. This is related to the Chernoff bound (see \cite{MR1995}), i.e., for all $\delta > 0$, it holds that
\begin{align}
  \Pr\left[\sum_{j=1}^rX_j > (1+\delta)\mu\right] < \frac{e^{\delta\mu}}{(1+\delta)^{(1+\delta)\mu}},
\end{align}
where $\mu = \sum_{j=1}^r\Pr[X_j = 1] = \frac{3}{2}$. Letting $h = (1+\delta)\mu$, we have
\begin{align}
  \Pr\left[\sum_{j=1}^rX_j > h\right] < \frac{e^{h-\mu}\mu^h}{h^h} \leq \frac{(e\mu)^h}{h^h} = \frac{(3e/2)^h}{h^h}.
\end{align}
Therefore, the probability of the Hamming weight being larger than $h$ is upper bounded by $\epsilon^2$ provided  
\begin{align}
  h \in O\left(\frac{\log(1/\epsilon)}{\log\log(1/\epsilon)}\right).
\end{align}
From this, we conclude that the occurrences of the multiplexed-$U$ gates can be reduced to $O\left(\frac{\log(1/\epsilon)}{\log\log(1/\epsilon)}\right)$ with error $\epsilon$.

  The number of qubits for indicator and purifier registers in a segment is still $O(r\log(mq))$. We use the similar compression scheme as in \cite{BCCKS14} to reduce the number of qubits for indicator and purifier registers. The intuition is to only store the positions of components with non-zero Hamming weight, and we also need two other registers to store the actual state in this position. 

  The compression scheme works as follows. We consider the initial sate $(\ket{0}\ket{\mu})^{\otimes r}$. After applying the multiplexed-$B$ gates (before applying the multiplexed-$U$ gates), the state becomes $(\mbox{multi-}B\ket{0}\ket{\mu})^{\otimes r}$. It can be written as a linear combination of basis states in the form 
  \begin{align}
    \ket{k_0,l_0}\cdots\ket{k_{r-1},l_{r-1}},
  \end{align}
 where $k_i\in\{0, \ldots, mq^2\}$ and $l_i\in\{0,\ldots,m\}$ for $i\in\{0,\ldots r-1\}$. For each basis state of Hamming weight at most $h$, define a tuple $g:=(g_0,\ldots,g_{h-1})$. We use $g_0,\ldots,g_{h-1}$ to represent the number of consecutive zero-Hamming weight components. For a basis state of Hamming weight $h' < h$, just set $g_{h'}=\cdots=g_{h-1} = r$. The state $\ket{g}$, together with two additional registers storing these components $(k_i, l_i)$ with non-zero Hamming weights, encode a basis state with Hamming weight at most $h$.

  To illustrate how the compression scheme works, consider a basis state 
  \[
	\ket{k_0,l_0}\cdots\ket{k_{r-1},l_{r-1}}
  \]
  where the positions of components with non-zero Hamming weights are $i_0, \ldots, i_{h-1}$. The Hamming weight of this basis state is $h$. Let $g:=(g_0,g_1,\ldots,g_{h-1})=(i_0, i_1-i_0-1, \ldots, i_{h-1}-i_{h-2}-1)$. Then the state $\ket{g}\ket{k_{i_0},\ldots, k_{i_{h-1}}}\ket{l_{i_0},\ldots,l_{i_{h-1}}}$ encodes the original basis state in the sense that it represents
  \begin{align*}
	\ket{0,0}^{\otimes g_0}\ket{k_{i_0}, j_{i_0}}\cdots\ket{0,0}^{\otimes g_{h-1}}\ket{k_{i_{h-1}},j_{i_{h-1}}}\ket{0,0}^{\otimes (r-h-i_0-\cdots-i_{h-1})}.
  \end{align*}

  In this encoding scheme, the register for $\ket{g}$ requires $O(\log(r)h)$ qubits. The two additional registers require $O(h\log(mq))$ qubits. If we prepare the indicator and purifier registers in this encoded representation, the number of qubits is
  \begin{align}
    O((\log(r)h + \log(mq)h)) = O((\log(1/\epsilon)+\log(mq))h).
  \end{align} 

  We briefly summarize this encoding scheme as follows. In the original representation, the initial state of the indicator and purifier registers is $(\ket{0}\ket{\mu})^{\otimes r}$, and we apply multiplexed-$B$ gates $\mbox{multi-}B^{\otimes r}$ on this state before applying multiplexed-$U$ gates. In the encoded representation, the initial state is 
$\ket{0^a}\ket{0^b}\ket{0^c}$, where
$a =  O(\log(r)h)$, $b =  O(\log(mq)h)$, and $c =  O(\log(m)h)$; 
the first and second registers correspond to the indicator register in the original representation, and the third register corresponds to the purifier register in the original representation. We denote the encoding operator by $E$. The operator $E$ corresponds to the multiplexed-$B$ gates in the original representation, as we apply $E$ on the encoded initial state before applying multiplexed-$U$ gates.

\subsection{Total number of gates and proof of the main theorem}
\label{section:gate-count}
In this section, we count the number of 1- and 2-qubit gates in our construction. There are three parts that we need to consider: the implementation of the encoding operator $E$, the implementation of the reflections in the oblivious amplitude amplification operator, and the implementation of the multiplexed-$U$ gates.
To complete the proof of the main theorem, all that remains is 
to bound the number of these gates.

  \begin{proof}[Proof of Theorem~\ref{thm:main_thm}]
We first consider the case where $t$ is as defined in Eq.~\eqref{eq:def-t}, so 
$t \pnorm{\textsf{pauli}}{\mathcal{L}} = \ln(2) + O(1/r)$.
The quantum circuit is based on the oblivious amplitude amplification operator $F$, whose correctness is shown by Lemma~\ref{lemma:OAA}. We modify the quantum circuit of $F$ by applying a concentration bound and the encoding scheme on the indicator and purifier registers as shown in subsection~\ref{section:hamming}. In the following, we show that this quantum circuit achieves the desired gate complexity.

For the encoding operator $E$, we first apply the techniques in \cite{BCG14} for $(\ket{0}\ket{\mu})^{\otimes r}$ to prepare the register $\ket{g}$. This can be done with $O\left(h(\log(r)+\log\log(1/\epsilon))\right) = O(\log(r)h)$ gates. In addition, we need to prepare a superposition of the basis states with non-zero Hamming weight in the second and third registers. This can be done with a slightly modified multiplex-$B$ gate and state $\ket{\mu}$ with gate cost $O(m^2q^2)$. Thus, the number of 1- and 2-qubit gates required for the encoding operator $E$ is $O(\log(r)h+m^2q^2) = O(\log(1/\epsilon)h+m^2q^2)$.

	In the oblivious amplitude amplification operator $F$, there are two reflections, $I-2P_0$ and $I-2P_1$, between $\widehat{W}$ and $\widehat{W}^{\dag}$. If we look into the constructions for $\widehat{W}$, the two reflections are between multiplexed-$B$ gates. To translate the operation $\left(\mbox{multi-}B^{\otimes r}\right)(I-2P_1)\left(\mbox{multi-}{B^{\dag}}^{\otimes r}\right)$ to the encoded representation, note that the multiplexed-$B$ gates correspond to the encoding operator $E$ in the encoded representation, and the reflection $I-2P_1$ is the reflection about the initial state $\ket{0}^{\otimes r}\ket{\mu}^{\otimes r}$. Hence in the encoded representation, the corresponding operation is first applying $E^{\dag}$, reflecting about the encoded initial state $\ket{0^a}\ket{0^b}\ket{0^c}$, where $a$, $b$, and $c$ are defined in the last paragraph of subsection~\ref{section:hamming}, and then applying $E$.

    A similar method applies to the operation $\left(\mbox{multi-}B^{\otimes r}\right)(I-2P_0)\left(\mbox{multi-}{B^{\dag}}^{\otimes r}\right)$. The only difference is that the reflection $I-2P_0$ is reflecting about the subspace where the state of the indicator register is $\ket{0}^{\otimes r}$. In the encoded representation, the corresponding reflection in the encoded representation should be about the subspace where the first two registers are in the state $\ket{0^a}\ket{0^b}$. Therefore, the corresponding operation in the encoded representation is first applying $E^{\dag}$, then applying the reflection about the encoded state $\ket{0^a}\ket{0^b}$ on the first two registers, and last applying $E$.

    The number of 1- and 2-qubit gates involved in the two reflections consists of the implementation of the encoding operator $E$, and two reflections. The number of gates for the reflections is of the same order of the number of qubits for the encoded representation. Therefore the number of 1- and 2-qubit gates in this part is $O(\log(1/\epsilon)h+m^2q^2)$.

	Each multiplexed-$U$ gate costs $O(mq^2(\log(mq)+n))$ of 1- and 2-qubit gates, as each controlled-$U$ requires $\log(mq)$ qubits for multiplexing and $O(n)$ Paulis, and we have to implement $O(mq^2)$ these controlled-$U$ gates. Since the number of occurrences of multiplexed-$U$ gates is $h$, the gate cost for this part is $O(mq^2h(\log(mq) + n))$

Therefore, the total number of 1- and 2-qubit gates is
\begin{align}
  O(m^2q^2+\log(1/\epsilon)h + mq^2(\log(mq)+n)h) \in 
  O\left(m^2q^2\frac{(\log(mq/\epsilon)+n)\log(1/\epsilon)}{\log\log(1/\epsilon)}\right).
\end{align}

For arbitrary evolution time $t$, let $\tau:=t\pnorm{\textsf{pauli}}{\mathcal{L}}$. Divide the evolution time into $O(\tau)$ segments. Then run this quantum circuit for a segment with precision $\epsilon/\tau$ and trace out the indicator and purifier registers. Repeat this $O(\tau)$ times and this evolution is simulated with total number of $1$- or $2$-qubit gates $O\left(m^2q^2\tau\frac{(\log(mq\tau/\epsilon)+n)\log(\tau/\epsilon)}{\log\log(\tau/\epsilon)}\right)$.

The distance between $\mathcal{N}$ and $e^{\mathcal{L}t}$ in terms of the diamond norm is established by Eqns.~\eqref{eq:dia-dist-1} and \eqref{eq:dia-dist-2}. Choosing $r$ large enough (i.e., $r=1/\epsilon$), the error of the simulation is within $\epsilon$. Note that the concentration bound and encoding scheme only cause $O(\epsilon)$ error.
\end{proof}



\section{Lindbladians with sparse Hamiltonian and Lindblad operators}
\label{sec:sparse}
In this section, we sketch the analysis of the simulation of Lindbladians with $d$-sparse Hamiltonian and Lindblad operators. Without loss of generality, we assume $\norm{H} \geq 1$ and $\norm{L_j} \geq 1$ for $j\in\{1,\ldots,m\}$. We first describe a method to approximate $H$ and $L_j$ as a linear combination of unitaries. Then we sketch the analysis of two key quantities: the normalized evolution time $t\pnorm{\textsf{pauli}}{\mathcal{L}}$, and the number of unitaries $q$ in this approximation

We consider the case where $H$ is $d$-sparse and each $L_j$ is both column and row $d$-sparse given by an oracle. Each $L_j$ can be decomposed as $L_j = \frac{L_j+L_j^{\dag}}{2}+i\frac{L_j-L_j^{\dag}}{2i}$, where $\frac{L_j+L_j^{\dag}}{2}$ and $\frac{L_j-L_j^{\dag}}{2i}$ are Hermitian. For each $\frac{L_j+L_j^{\dag}}{2}$, $\frac{L_j-L_j^{\dag}}{2i}$, and $H$, we use the methods in \cite{BCCKS14} to approximate them as a linear combination of unitaries with equal coefficient  $\gamma$. The error of the approximation is $O(d^2\gamma)$ in terms of the max norm, and the number of unitaries in the approximation is $O\bigl(d^2\pnorm{{\textsf{max}}}{H}/\gamma\bigr)$ for $H$ and $O\left(d^2\pnorm{{\textsf{max}}}{L_j}/\gamma\right)$ for $L_j$.
It is easy to see that $c_0 = O\left(d^2\pnorm{{\textsf{max}}}{H}\right)$ and $c_j = O\left(d^2\pnorm{{\textsf{max}}}{L_j}\right)$ for $j\in\{1,\ldots,m\}$. We have $t\pnorm{\textsf{pauli}}{\mathcal{L}} = t\bigl(c_0+\sum_{j=1}^mc_j^2\bigr) = O\bigl(td^4\pnorm{\textsf{ops}}{\mathcal{L}}\bigr)$.

To bound $q$, which is the number of terms in the LCU decomposition for each of $H$ and $L_j$, we consider the error of this approximation. As each of $H$ and $L_j$ can be approximated with error $O(d^2\gamma)$ in terms of the max norm, $\mathcal{L}$ can be approximated with error $O\bigl(d^3\pnorm{\textsf{ops}}{\mathcal{L}}\gamma\bigr)$ in terms of the diamond norm. To restrict the simulation error within $\epsilon$ for evolution time $t$, $\gamma$ can be chosen so that $\gamma = O\bigl(\epsilon/(td^3\pnorm{\textsf{ops}}{\mathcal{L}})\bigr)$. Therefore, the number of unitaries in the decomposition is bounded by $q=O\bigl(td^5\pnorm{\textsf{ops}}{\mathcal{L}}^2/\epsilon\bigr)$.

In the implementation of the multiplexed-$U$ gates, we no longer need to implement all the $O(mq^2)$ unitaries, since the oracles for $H$ and $L_j$ are given. Also, the cost for implementing the encoding scheme becomes $O(\log(1/\epsilon)h + \poly(n))$ as the coefficients in the LCU for each $H$ and $L_j$ are the same, which saves the $O(m^2q^2)$ factor. Thus the $m$ and $q$ factors in the gate complexity can be eliminated. (The $\log(m)$ and $\log(q)$ factors will be preserved.) Let $\tau=t\pnorm{\textsf{ops}}{\mathcal{L}}$. By our construction, the gate complexity is 
\[
  O\bigl(\tau\,\polylog(mq\tau/\epsilon)\,\poly(n,d)\bigr).
\]

The query complexity is the number of occurrences of the multiplexed-$U$ gates, which is
\[
  O\Biggl(\tau\,\frac{\log(\tau/\epsilon)}{\log\log(\tau/\epsilon)}\,\poly(d)\Biggr).
\]

\section{Acknowledgments}
We thank Andrew Childs, Patrick Hayden, Martin Kliesch, Tongyang Li, Hans Massen, Barry Sanders, and Rolando Somma for helpful discussions. This research was supported in part by Canada's NSERC and an NSERC Canada Graduate Scholarship (Doctoral).

\bibliography{ref}
\bibliographystyle{plain}

\appendix
\section{Cost of expressing Lindblad evolution as Hamiltonian evolution}
\label{appendix:total-is-infinite}


Let $\mathcal{L}$ be a Lindbladian acting on an $n$-qubit register $\mathcal{H}$ over a time interval $[0,T]$.
For each initial state, $\mathcal{L}$ associates a \textit{trajectory}, consisting of a density operator $\rho(t)$ for each $t \in [0,T]$.
Here we show that if this is simulated by Hamiltonian evolution in a larger system with an ancillary register that is continually reset (expressed as a limiting case when $N \rightarrow \infty$ in the process illustrated in Fig.~\ref{fig:interleave-LB}) then the total evolution time for this Hamiltonian can be necessarily infinite.

\begin{definition}
Define an \emph{$N$-stage $\epsilon$-precision discretization of $\mathcal{L}$ for interval $[0,T]$} as an ancillary register $\mathcal{K}$, a Hamiltonian $H$ (with $\|H\| = 1$) acting on the joint system $\mathcal{K}\otimes \mathcal{H}$, and $\delta \ge 0$ such that the channel $\mathcal{N}_{H\delta}$ defined as 
\begin{align}
\mathcal{N}_{H\delta}[\rho] = \mathrm{Tr}_{\mathcal{K}}\bigl(e^{-iH\delta}(\ketbra{0}{0}\otimes\rho)e^{iH\delta}\bigr)
\end{align}
has the following property.
$\mathcal{N}_{H\delta}$ approximates evolution under $\mathcal{L}$ in the sense that, for each $k \in \{1,\dots,N\}$,
\begin{align}
\bigl\|(\mathcal{N}_{H\delta})^k - \exp\bigl(\textstyle{\frac{kT}{N}}\mathcal{L}\bigr) \bigr\|_{\diamond} \le \epsilon.
\end{align}
That is, the $N$ points generated by $\mathcal{N}_{H\delta}, (\mathcal{N}_{H\delta})^2,\dots,(\mathcal{N}_{H\delta})^N$ approximate the corresponding points on the trajectory determined by $\mathcal{L}$.
\end{definition}
\begin{figure}[!ht]
\centering
\includegraphics[width=0.90\textwidth]{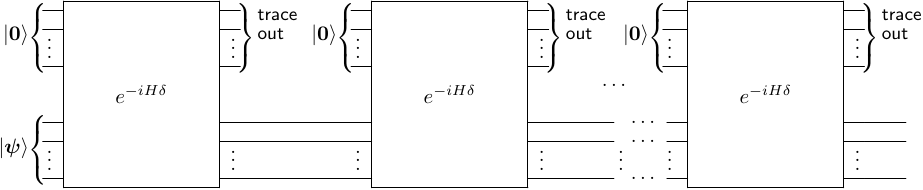}
\caption{\small $N$-stage $\epsilon$-precision discretization of the trajectory resulting from $\mathcal{L}$. For each $k \in \{1,\dots,N\}$, after $k$ stages, the channel should be within $\epsilon$ of $\exp\bigl(\frac{kT}{N}\mathcal{L}\bigr)$.}\label{fig:interleave-LB}
\end{figure}

Our lower bound is for the \emph{amplitude damping process} on a 1-qubit system which is the time-evolution described by the Lindbladian $\mathcal{L}$, where
\begin{align}
\mathcal{L}[\rho] = L\rho L^{\dag} -\textstyle{\frac{1}{2}}(L^{\dag}L\rho + \rho L^{\dag}L),
\end{align}
and 
$\displaystyle{L = \begin{pmatrix}
 0 & 1 \\ 0 & 0
\end{pmatrix}}$.

\begin{theorem}\label{thm:Hamiltonian-cost-AD}
Any $\frac{1}{4}$-precision $N$-stage discretization of the amplitude damping process over the time interval $[0,\ln 2]$ has the property that the total evolution time of $H$ is $\Omega(\sqrt{N})$.
(Note that this lower bound is independent of the dimension of the ancillary system.)
\end{theorem}

To prove Theorem~\ref{thm:Hamiltonian-cost-AD}, we first prove the following \emph{Local Hamiltonian Approximation} lemma.
This concerns a scenario where $H$ is a Hamiltonian acting on a joint system of two registers, a system register $\mathcal{H}$ and an ancillary register 
$\mathcal{K}$, and where $\mathcal{K}$ is traced out after this evolution.
Informally, the lemma states that, if the initial state is a product state and the evolution time is short, then this process can  be approximated by the evolution of another Hamiltonian $G$ that acts on $\mathcal{H}$ alone.
This is illustrated in Fig.~\ref{fig:local-hamiltonian-approx}.
\begin{figure}[ht!]
\centering
\includegraphics[width=0.80\textwidth]{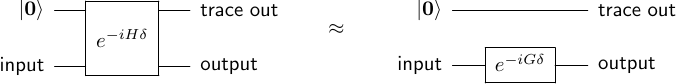}\caption{The Local Hamiltonian Approximation Lemma. The first register is $d$-dimensional, the second register contains $n$ qubits, and the approximation is within $O(\delta^2)$ (independent of $d$ and $n$).}\label{fig:local-hamiltonian-approx}
\end{figure}

\begin{lemma}[Local Hamiltonian approximation]\label{lem:local-unitary-approximation}
Let $\mathcal{H}$ be an $n$-qubit register and $\mathcal{K}$ a $d$-dimensional register.
Let $H$ be a Hamiltonian (with $\|H\|=1$)  acting on the joint system $\mathcal{K}\otimes\mathcal{H}$.
Define the $n$-qubit channel $\mathcal{N}_{H\delta}$ as
\begin{align}
\mathcal{N}_{H\delta}[\rho] = \mathrm{Tr}_{\mathcal{K}}\bigl(e^{-iH\delta}(\ketbra{0}{0}\otimes\rho)e^{iH\delta}\bigr).
\end{align}
Then there exists a Hamiltonian $G$ (with $\|G\|=1$), acting on $\mathcal{H}$ alone, such that $\mathcal{N}_{G\delta}$ defined as 
\begin{align}
\mathcal{N}_{G\delta}[\rho] = e^{-iG\delta}\rho \, e^{iG\delta}
\end{align}
satisfies
$\| \mathcal{N}_{H\delta} - \mathcal{N}_{G\delta} \|_1 \in O(\delta^2)$.
(The notation $\|\cdot\|_1$ indicates the induced trace norm, which is sufficient for our purposes because our application is a lower bound.)
\end{lemma}

\begin{proof}
Viewing $H$ as a $d \times d$ block matrix, we have
\begin{align}
H = \sum_{j=0}^{d-1} \sum_{k=0}^{d-1} \ketbra{j}{k}\otimes H_{jk}
\end{align}
and we refer to $H_{jk}$ as the \textit{$(j,k)$-block}.
Define $D$ as the diagonal blocks of $H$, namely
\begin{align}
D = \sum_{j=0}^{d-1} \ketbra{j}{j}\otimes H_{jj},
\end{align}
and set $J = H - D$ (the off-diagonal blocks).
Note that $\|D\| \leq 1$, $\|J\| \leq 2$, and $\|e^{-iH\delta} - e^{-iD\delta} e^{-iJ\delta}\| \le \delta^2$, for $\delta > 0$, which permits us to consider the effect of $J$ and $D$ separately.

Now consider the state $e^{-iJ\delta}\ket{0}\otimes\ket{\psi}$.
We will show that, if the measurement corresponding to projectors $\ketbra{0}{0}$ and 
$I-\ketbra{0}{0}$ is performed on register $\mathcal{K}$, then the residual state has trace distance $O(\delta^2)$ from $\ket{0}\otimes\ket{\psi}$.
Since the $(0,0)$-block of $J$ is $0$, 
\begin{align}
J\delta\, \ket{0}\otimes\ket{\psi} = \delta'\ket{\Psi^\perp}, 
\end{align}
where $\ket{\Psi^\perp}$ is a state such that $(\ketbra{0}{0}\otimes I)\ket{\Psi^\perp} = 0$ and $0\le \delta' \le \delta$.
Therefore,
\begin{align}
e^{-iJ\delta}\ket{0}\otimes\ket{\psi} 
&=
\sum_{r=0}^{\infty} \frac{(-iJ\delta)^r}{r!} \ket{0}\otimes\ket{\psi}\\
&= \ket{0}\otimes\ket{\psi} - i\delta'\ket{\Psi^\perp} + \delta''\ket{\Phi},
\end{align}
where $0 \le \delta'' \le e^{\delta}-1-\delta \in O(\delta^2)$.
It follows that, if the above measurement is performed on register $\mathcal{K}$, then the probability of measurement outcome $I - \ketbra{0}{0}$ is at most 
$(\delta')^2 + (\delta'')^2 \in O(\delta^2)$.
This implies that the state when register $\mathcal{K}$ of $e^{-iJ\delta}\ket{0}\otimes\ket{\psi}$ is traced out, namely
\begin{align}
\Tr_{\mathcal{K}}\bigl(e^{-iJ\delta}(\ketbra{0}{0}\otimes\ketbra{\psi}{\psi})e^{iJ\delta}\bigr),
\end{align}
has trace distance $O(\delta^2)$ from the original state $\ketbra{\psi}{\psi}$.

Therefore, for states of the form $\ket{0}\otimes\ket{\psi}$, the operation $e^{-iH\delta}$ can be approximated by $e^{-iD\delta}$ at the cost of an error of $O(\delta^2)$ in trace distance.
The result follows by setting $G = H_{00}$ (the $(0,0)$-block of $D$).
\end{proof}

\begin{proof}[Proof of Theorem~\ref{thm:Hamiltonian-cost-AD}]
It is straightforward to check that, starting with the initial state $\ketbra{1}{1}$ and evolving by the amplitude 
damping process for time $T = \ln 2$ produces the maximally mixed state.

Consider any $\frac{1}{4}$-precision $N$-stage discretization of this process, with Hamiltonian $H$ and $\delta > 0$. We can apply the Local Hamiltonian Approximation Lemma (Lemma~\ref{lem:local-unitary-approximation}) to approximate each of the $N$ evolutions of $H$ with evolution by a Hamiltonian $G$ that is local to the system register.
The result is unitary evolution of the qubit that approximates the amplitude damping process within trace distance error at most $O(N\delta^2)$.

Unitary evolution applied to $\ketbra{1}{1}$ results in a pure state, 
and the trace distance between any pure state and the maximally mixed state is $\frac{1}{2}$.
Therefore, to avoid a contradiction, we must have $N\delta^2 \in \Omega(1)$, which implies that $\delta \in \Omega(1/\sqrt N)$.
Therefore, the total evolution time of $H$ is $N \delta \in \Omega(\sqrt N)$.
\end{proof}


\section{Proof that $\|e^{\delta\mathcal{L}} - (\mathbbm{1} + \delta\mathcal{L})\|_{\diamond} \le (\delta\|\mathcal{L}\|_{\diamond})^2$ for small $\delta$}\label{app:1plusL}

Assume that $0 \le \delta \|\mathcal{L}\|_{\diamond} \le 1$.
Then, for any $X$ such that $\| X \|_1 \le 1$, 
\begin{align*}
\bigl\|(e^{\delta\mathcal{L}} - (\mathbbm{1} + \delta\mathcal{L}))[X]\bigr\|_1
&=
\Biggl\|\sum_{s=2}^{\infty} \frac{\delta^s}{s!}\mathcal{L}^{(s)}[X]\Biggr\|_1 \\
&\le 
\sum_{s=2}^{\infty} \frac{\delta^s}{s!}\bigl\|\mathcal{L}^{(s)}[X]\bigr\|_1 \\
&\le 
\sum_{s=2}^{\infty} \frac{\delta^s}{s!}\bigl(\|\mathcal{L}[X]\|_1\bigr)^s \\
&\le \bigl(\delta\|\mathcal{L}[X]\|_1\bigr)^2 \\
&\le \bigl(\delta\|\mathcal{L}\|_1\bigr)^2 ,
\end{align*}
where we are using the fact that $e^z - (1+z)\le z^2$ when $0 \le z \le 1$.

To extend this from the induced trace norm to the diamond norm, note that, for two registers $\mathcal{H}$ and $\mathcal{K}$, 
\begin{align*}
(e^{\delta\mathcal{L}} - (\mathbbm{1}_{\mathcal{H}} + \delta\mathcal{L}))\otimes \mathbbm{1}_{\mathcal{K}}
&=
e^{\delta(\mathcal{L}\otimes \mathbbm{1}_{\mathcal{K}})} - (\mathbbm{1}_{\mathcal{HK}}
 + \delta(\mathcal{L}\otimes \mathbbm{1}_{\mathcal{K}}))
\end{align*}
and, $\mathcal{L}\otimes \mathbbm{1}_{\mathcal{K}}$ is a Lindbladian with
$\|\mathcal{L} \otimes \mathbbm{1}_{\mathcal{K}}\|_1 = \|\mathcal{L}\|_{\diamond}$ when the dimensions of $\mathcal{H}$ and $\mathcal{K}$ are equal.
This implies
\begin{align*}
\Bigl\|e^{\delta\mathcal{L}} - (\mathbbm{1}_{\mathcal{H}} + \delta\mathcal{L})\Bigr\|_{\diamond}
&=
\Bigl\|\bigl(e^{\delta\mathcal{L}} - (\mathbbm{1}_{\mathcal{H}} + \delta\mathcal{L})\bigr)\otimes \mathbbm{1}_{\mathcal{K}}\Bigr\|_1 \\ 
&=
\Bigl\|e^{\delta(\mathcal{L}\otimes \mathbbm{1}_{\mathcal{K}})} - (\mathbbm{1}_{\mathcal{HK}}
 + \delta(\mathcal{L}\otimes \mathbbm{1}_{\mathcal{K}}))\Bigr\|_1 \\
&\le 
\bigl(\delta\|\mathcal{L}\otimes \mathbbm{1}_{\mathcal{K}}\|_1)^2 \\
&= 
\bigl(\delta\|\mathcal{L}\|_{\diamond})^2.
\end{align*}


\end{document}